\documentclass{article}

\usepackage{mathptmx}       
\usepackage{helvet}         
\usepackage{courier}        
\usepackage{graphicx}        
\usepackage{multicol}        
\usepackage{amsmath,amsfonts}
\usepackage{enumerate}
\usepackage{geometry}
\newcommand{\HKRnextcol}{\ensuremath{\operatorname{nextcol}}}
\newcommand{\HKRnextrow}{\ensuremath{\operatorname{nextrow}}}

\newtheorem{theorem}{Theorem}

\newtheorem{lemma}{Lemma}
\newtheorem{proposition}{Proposition}
\newtheorem{corollary}{Corollary}
\newenvironment{proof}
{\begin{trivlist}\item[]{{\sc
Proof.}}}{\hfill{$\square$}\noindent\end{trivlist}
}

\begin{document}

\title{The Maximum Scatter TSP on a Regular Grid}
\author{Isabella Hoffmann, Sascha Kurz, and J\"org Rambau\\
Universit\"at Bayreuth, 95440 Bayreuth\\
\texttt{firstname.lastname@uni-bayreuth.de}}
\maketitle
\abstract{In the maximum scatter traveling salesman problem the
  objective is to find a tour that maximizes the shortest distance
  between any two consecutive nodes. This model can be applied to
  manufacturing processes, particularly laser melting processes.  We
  extend an algorithm by Arkin et al. that yields optimal solutions
  for nodes on a line to a regular $m \times n$-grid. The new algorithm
  $\textsc{Weave}(m,n)$ takes linear time to compute an optimal tour in
  some cases. It is asymptotically optimal and a $\frac{\sqrt{10}}{5}$-approximation
   for the $3\times 4$-grid, which is the worst case.}

\section{Introduction}
\label{sec:1}
The Maximum Scatter Traveling Salesman Problem (Maximum Scatter TSP)
asks for a closed tour in a graph visiting each node exactly once so
that its shortest edge is as long as possible. Applications,
approximation strategies, and exact algorithms for the geometric
version where all vertices and edges are on a line or on a circle were
presented by Arkin et al.~\cite{ACM}. Moreover, they proved that the problem 
is NP-complete in general. Chiang~\cite{Chiang} focuses on the max-min 
$m$-neighbor TSP (particularly with $m=2$), which we do not consider here. 
It is not known whether the geometric Maximum Scatter Problem for points in the 
plane is NP-hard. In this paper we extend their strategy on the line to the 
geometric Maximum Scatter TSP in the plane where all points are located on a 
rectangular equidistant grid with $m$ rows and $n$ columns. More formally, a
regular rectangular grid is defined as the complete graph $G(m,n)$ on
the vertex set
$V(m,n)=\{\tbinom{x}{y} \in \mathbb{Z}^2|1\leq x\leq n,1\leq y \leq
m\}$.
We call the problem $\textsc{MSTSP}(m,n)$ and our new algorithm
$\textsc{Weave}(m,n)$. 

All current results are summarized in Theorem \ref{thm:main_thm}. 

\begin{theorem}
  Let $m\leq n$ be natural numbers, and let
  $k =\bigl\lfloor\frac{n}{2}\bigr\rfloor$ and
  $t =\bigl\lfloor\frac{m}{2}\bigr\rfloor$. Then there is a
  linear-time algorithm $\textsc{Weave}(m,n)$ that specifies a
  feasible tour for $\textsc{MSTSP}(m,n)$ that is
  \begin{enumerate}[(i)]
  \item optimal whenever $n$ is odd, $m=n$, or $m=2$;
  \item a $(\sqrt{1-\frac{2(k-t)}{(k-t)^2-2t(k-1)+1}})$-approximation
    whenever $m$ and $n$ are even;
  \item a $(\sqrt{1-\frac{2k-1}{k^2+t^2}})$-approximation whenever $n$
    is even and $m$ is odd.
  \end{enumerate}
  \label{thm:main_thm}
\end{theorem}
The algorithm is in all cases no worse than a 
$\frac{\sqrt{10}}{5}$-approximation. More specifically, it will turn
out that the additive gap between the objective value of
$\textsc{Weave}(m,n)$ and the optimal value of $\textsc{MSTSP}(m,n)$
is always strictly smaller than one, the length of a shortest edge in
$G(m,n)$.

\section{The Algorithm}
$\textsc{Weave}(m,n)$ is based on an algorithm for a tour for points
on a line introduced by Arkin et al.~\cite{ACM}. The resulting order of
 the points on a line is exactly transferred to the columns in the equidistant
grid. That is why we will consider this order here in more detail. 
    
There are two subroutines in the algorithm, one for an odd number of
points and one for an even number of points. For an odd number of
points $n=2k+1\ge 3$, the distances between subsequent points in the order are
either $k$ or $k+1$, see the left two parts of
Fig.~\ref{fig:order_line}. So the order of the points in the tour is:
$1,k+2,2,k+3,\ldots,n,k+1$. The tour is completed by returning to the
first point. For further use, we call this procedure
$\textsc{LineOdd}(n)$. 
\begin{figure}
  \includegraphics[scale=0.4]{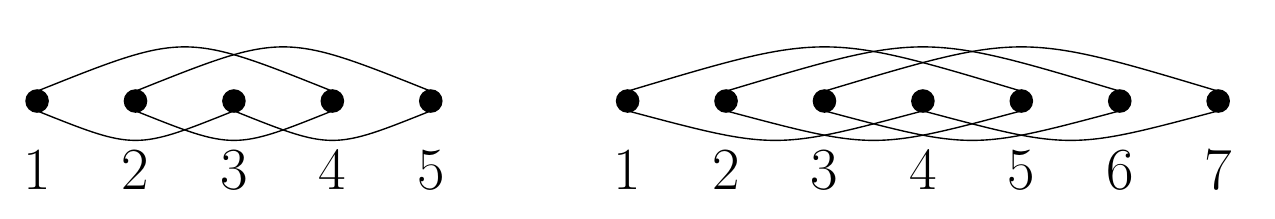} \quad
  \includegraphics[scale=0.4]{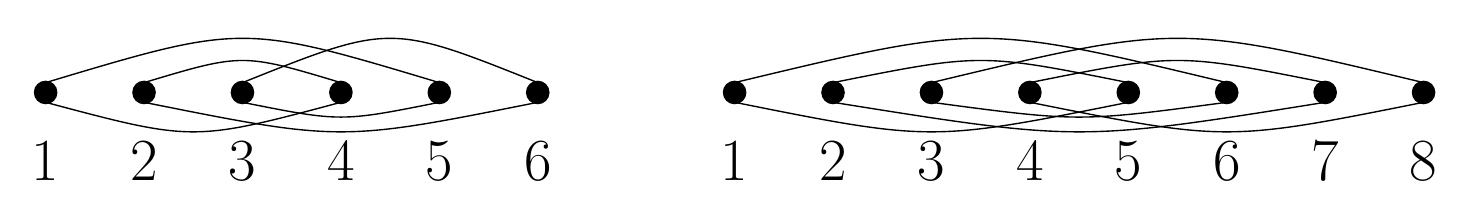} 
  \caption{Order of $\textsc{LineOdd}(n)$ and $\textsc{LineEven}(n)$ for n=5,6,7,8}
  \label{fig:order_line}
\end{figure}
For an even number of points $n=2k\ge 4$, the distances are chosen to be
alternatingly $k-1$ and $k+1$. At the points at the end of the line and at the
 two closest points to the center of the line the distance of $k$ is used as well. So the
distance of $k$ is filled in between the alternating distances when
reaching the endpoints or the center, see the last two components of
Fig.~\ref{fig:order_line}. Equally one can say, the end points have
distances of $k$ and $k+1$ and the two points around the center have
distances of $k-1$ and $k$. This procedure is called
$\textsc{LineEven}(n)$ in the following.  If $k$ is odd, the order of
the points in the tour is
$1,k+2,3,k+4,\ldots,n-1,k,n,k-1,\ldots,k+1$. 
If $k$ is even, the order
of the points in the tour is
$1,k+2,3,k+4,\ldots,k-1,n,k,n-1,\ldots,k+1$. 
From Arkin et
al.~\cite{ACM} (Sec. 6.1 and 6.2) it follows 
that the tours generated by $\textsc{LineOdd}(n)$ and
$\textsc{LineEven}(n)$ are optimal tours for 
$\textsc{MSTSP}(1, n)$.  Thus, $\textsc{Weave}(1,n)$ is set to 
$\textsc{LineOdd}(n)$ or $\textsc{LineEven}(n)$, respectively.

In the following we present the new algorithm for $m \ge 2$.  The
points in the grid are labeled by $(i,j)$, where $i = 1, 2, \dots, m$
is the row index and $j=1,\ldots , n$ is the column index.  The simple
idea to apply $\textsc{Weave}(1,n)$ row by row yields $k-1$ as the
shortest edge length. However, this can be improved by using the
freedom in the second dimension. One can increase the length of the
shortest edge in $\textsc{Weave}(1,n)$ by using the order of nodes in
both dimensions at the same time. However, this yields no closed tour.
Thus, we need a new pattern in the second dimension.

To this end, the rows are partitioned in pairs for $n$ even and pairs
plus one triple for $n$ odd.  The paired rows are at distance~$t$, as
are the rows in the triple. The idea is to jump back and forth inside
a pair (triple) of rows where the order of columns is given by
$\textsc{Weave}(1,n)$. The successor node for a node always lies in 
the paired (tripled) row to the current node.  For example, for odd $n$ 
this is done until all nodes in the pair have been visited, whereas for 
even $n$ this is done for half of the nodes. Then we traverse into the 
next pair.  Since switching to the next pair yields a smaller vertical 
distance, we start in a column such that these pair-switches are 
accompanied by the largest possible horizontal distances $k+1$ 
in $\textsc{Weave}(1,n)$. Thus, $\textsc{Weave}(m,n)$ starts each pair 
in column $k+2$. The setup is illustrated in Fig.~\ref{fig:princ_alg}. 
\begin{figure}
  \centering
  \includegraphics[scale=0.55]{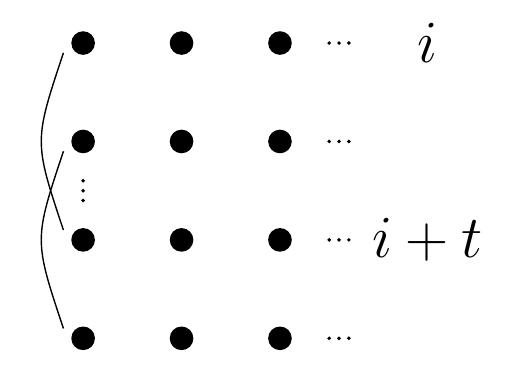} \qquad
  \includegraphics[scale=0.55]{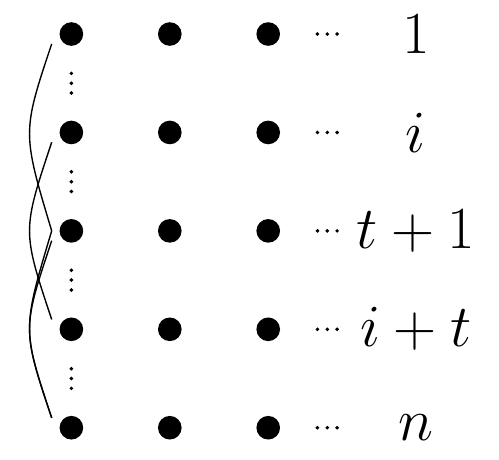} \qquad
  \includegraphics[scale=0.55]{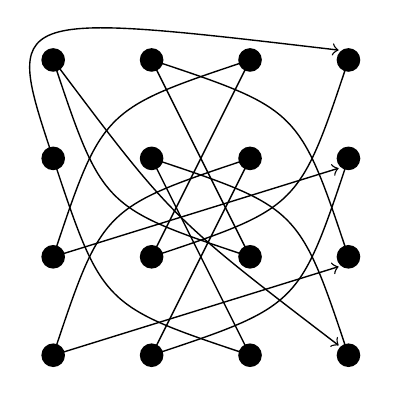} 
  \caption{Pairs of rows for an even number, for an odd number of
    rows, and an example of a ($4\times 4$)-grid}
\label{fig:princ_alg}
\end{figure}
 
We first describe $\textsc{Weave}(m,n)$ in a way that the construction
principle becomes apparent.  First, let $m$ be even.

For an odd number of columns the algorithm works in the following way,
where in the following the order of column indices in each line is
given by $\textsc{Weave}(1,n)$:
\begin{align}
  (1,k+2),(t+1,2),(1,k+3),(t+1,3), &\ldots,(1,n),(t+1,k+1),(1,1),\notag\\
  (t+1,k+2),(1,2),(t+1,k+3),(1,3), &\ldots,(t+1,n),(1,k+1),(t+1,1),\notag\\
  (2,k+2),(t+2,2),(2,k+3),(t+2,3), &\ldots,(2,n),(t+2,k+1),(2,1),\notag\\
  (t+2,k+2),(2,2),(t+2,k+3),(2,3), &\ldots,(t+2,n),(2,k+1),(t+2,1),\notag\\
                                   &\ldots,\notag\\
  (t,k+2),(m,2),(t,k+3),(m,3), &\ldots,(t,n),(m,k+1),(t,1),\notag\\
  (m,k+2),(t,2),(m,k+3),(t,3), &\ldots,(m,n),(t,k+1),(m,1),(1, k+2).
\end{align}

For an even number of columns this procedure fails because it
generates a subtour inside a pair of rows.  This can be rectified by
switching pairs after each completed order of $\textsc{Weave}(1,n)$.
The resulting order of nodes is as follows (again, the order of
columns in each line is given by $\textsc{Weave}(1,n)$):
\begin{align}
  (1,k+2),(t+1,3),(1,k+4),(t+1,5), &\ldots,(t+1,2),(1,k+1),(t+1,1),\notag\\
  (2,k+2),(t+2,3),(2,k+4),(t+2,5),&\ldots,(t+2,2),(2,k+1),(2,1),\notag\\
                                  &\ldots,\notag\\
  (t+1,k+2),(1,3),(t+1,k+4),(1,5),&\ldots,(t+1,2),(1,k+1),(t+1,1),\notag\\
  (t+2,k+2),(2,3),(t+2,k+4),(2,5),&\ldots,(t+2,2),(2,k+1),(t+2,1),\notag\\
                                  &\ldots,\notag\\
  (m,k+2),(t,3),(m,k+4),(t,5),&\ldots,(m,2),(t,k+1),(m,1),(1,k+2).
\end{align}

Let now $m$ be odd. First, the triple $(1,t+1,m)$ of rows is
taken. 
If $n$ is not divisible by $3$, then the rows are visited in the order
$(1,t+1,m,1,t+1,m,\ldots)$, $(t+1,m,1,t+1,m,1,\ldots)$,
$(m,1,t+1,m,1,t+1,\ldots)$ or $(1,t+1,m,1,t+1,m,\ldots)$,
$(m,1,t+1,m,1,t+1,\ldots)$, $(t+1,m,1,t+1,m,1,\ldots)$, depending
on~$n$.  During this process, the columns are visited in the order of
$\textsc{Weave}(1,n)$ three times consecutively. 

If $n$ is a multiple of~$3$ the row order pattern is always
$(1,t+1,m,1,t+1,m,\ldots)$, $(m,1,t+1,m,1,t+1,\ldots)$,
$(t+1,m,1,t+1,m,1,\ldots)$. After the row triple has been traversed
completely, the procedure is the same as for even~$m$.

We now define for $\textsc{Weave}(m,n)$ the successor node
$\bigl(\HKRnextrow(i,j), \HKRnextcol(i,j)\bigr)$ of each node $(i,j)$
in every possible case. 
Together with the starting node $(1,k+2)$ this yields a complete definition of our
algorithm. Recall that $t= m-t$ if $m$ is even and $t= m-1-t$ if $m$ is odd. 

We start with the successor column of $\HKRnextcol(i,j)$ when $n$ is
odd:
\begin{equation*}
  \HKRnextcol(i,j) := \begin{cases}  
    j+k+1&\text{if } j\leq k \\  
    j-k &\text{if } j> k. 
  \end{cases}
\end{equation*}
The successor column when $n$ is even and $k$ is even:
\begin{equation*}
  \HKRnextcol(i,j) := \begin{cases}  j+k+1&\mbox{if } j\leq k \mbox{ and } j \ \mbox{ odd} \\ 
    j+k-1&\mbox{if } j\leq k \mbox{ and } j \mbox{ even} \\ 
    j-(k+1)&\mbox{if }n> j> k+1 \mbox{ and } j \ \mbox{ odd} \\ 
    j-(k-1)&\mbox{if } n> j> k+1 \mbox{ and } j \mbox{ even} \\ 
    j-k&\mbox{if }	j=k+1 \mbox{ or } j=n,
  \end{cases}
\end{equation*}
The successor column when $n$ is even and $k$ is odd: 
\begin{equation*}
  \HKRnextcol(i,j) := \begin{cases}  j+k+1&\mbox{if } j< k \mbox{ and } j \  \mbox{ odd} \\ 
    j+k-1&\mbox{if } j< k \mbox{ and } j \mbox{ even} \\ 
    j-(k+1)&\mbox{if } j> k+1 \mbox{ and } j \mbox{ even} \\ 
    j-(k-1)&\mbox{if } j> k+1 \mbox{ and } j \ \mbox{ odd} \\ 
    j+k&\mbox{if } j=k \\
    j-k&\mbox{if }j=k+1.
  \end{cases}
\end{equation*}
Let now $m,n >3$.
The successor row $\HKRnextrow(i,j)$ of $(i,j)$ when $n$ is odd [even]
and $m$ is even:
\begin{equation*}
  \HKRnextrow(i,j) := \begin{cases} 
    i+t &\mbox{if } i\leq t  \bigl[\mbox{and } j\neq 1\bigr]  \\
    i-t &\mbox{if } i> t \mbox{ and }  j\neq 1 \\
    \bigl[i+(t+1)&\mbox{if } i< t \mbox{ and } j=1 \bigr] \\ 
    i-(t-1) &\mbox{if } m>i> t \mbox{ and } j=1 \bigl[\mbox{or } (i=m \mbox{ and } j=1)\bigr] \\
    \mbox{TERMINATE} &\mbox{if } i=m\bigl[-t \bigr] \mbox{ and } j=1\\ 
  \end{cases}
\end{equation*}
The successor row when $n$ is odd [even], $m$ is odd, and $n \bmod{3}\neq 0$:
\begin{equation*}
  \HKRnextrow(i,j) := \begin{cases} 
    i+t &\mbox{if } i\leq t+1 \bigl[\mbox{and } j \neq 1 \mbox{ or } \\
    &((i=1 \mbox{ or } i=t+1) \mbox{ and }j= 1)\bigr]\\
    i-t &\mbox{if } m>i> t+1 \mbox{ and }  j\neq 1 \\
    1&\mbox{if } i=m \mbox{ and } j\neq 1 \\ 
    \bigl[i+(t+1)&\mbox{if } 1< i< t \mbox{ and } j= 1 \bigr] \\ 
     i-(t-1)&\mbox{if } m-1>i> t+1\mbox{ and }  j= 1\\ 
     [i-(t-2)&\mbox{if } i=m-1 \mbox{ and }  j= 1]\\
    2&\mbox{if } i=m \mbox{ and } j=1 \\
    \mbox{TERMINATE} &\mbox{if } i=m-1\bigl[-t\bigr] \mbox{ and } j=1
  \end{cases}
\end{equation*}
The successor row when $n$ is odd [even], $m$ is odd, and $n \bmod{3}= 0$:
\enlargethispage{2cm}
\begin{equation*}
  \HKRnextrow(i,j) := \begin{cases} 
    i+t &\mbox{if } 1<i\leq t  \bigl[\mbox{and } j\neq 1\bigr]\\ 
    i+t &\mbox{if } (i=1 \mbox{ or } i= t+1) \mbox{ and } j\neq 1 \\
    i-t &\mbox{if } m>i> t+1 \mbox{ and }  j\neq 1 \\
    1&\mbox{if } i=m \mbox{ and } j\neq 1 \\ 
    m &\mbox{if } i=1  \mbox{ and } j=1 \\
    t+1 &\mbox{if } i=m  \mbox{ and } j=1 \\
    \bigl[i+(t+1)&\mbox{if } 1< i< t \mbox{ and }  j= 1\bigr] \\
    i-(t-1)&\mbox{if } m-1>i\geq t+1 \mbox{ and } j= 1\\ 
    [i-(t-2)&\mbox{if } i=m-1 \mbox{ and }  j= 1]\\ 
    \mbox{TERMINATE} &\mbox{if } i=m-1\bigl[-t\bigr] \mbox{ and } j=1.
  \end{cases}
\end{equation*}
If $m=3$, then $\textsc{Weave}(m,n)$ terminates when for the first
time the triple of rows is finished.

\section{Optimality and Gaps}
\label{sec:3}
We now analyze the edge lengths appearing in the tours of the
algorithm, in particular, the shorter ones to get a lower bound for a
solution to the Maximum Scatter TSP. After that, we have a look at the
upper bounds for the longest possible shortest edge in a
tour. 

By comparing upper bounds and lower bounds we can prove that 
$\textsc{Weave}(m,n)$ is optimal for any grid with an odd number of 
columns, for any quadratic grid and for any grid with two rows. 
In the following, we identify the Euclidean lengths of candidates for the
 shortest edge in a tour computed by $\textsc{Weave}(m,n)$.
 
For odd $n$ the relevant edge lengths are $\sqrt{k^2+t^2}$ and 
$\sqrt{(k+1)^2+(t-1)^2}$. For $n$ even the lengths are $\sqrt{t^2+(k-1)^2}$
 and $\sqrt{(t-1)^2+(k+1)^2}$. A straight-forward
calculation shows: $t^2+k^2 \leq (t-1)^2+(k+1)^2$ and
$t^2+(k-1)^2 \leq (t-1)^2+(k+1)^2$, so the lengths of the shortest
edges of $\textsc{Weave}(m,n)$ are $\sqrt{k^2+t^2}$ for $n$ odd and
$\sqrt{t^2+(k-1)^2}$ for $n$ even. 

An upper bound for the shortest edge can be derived by considering the
distances between the nodes in the center of the grid and the corner
nodes. Each node has to be connected to exactly two other nodes in a
tour.  This holds, in particular, for the most central nodes. The
nodes with the longest distances to them are the corners of the
grid. If $m$
and $n$
are odd, there is exactly one central node, and the distance to each
of the four corners is $\sqrt{k^2+t^2}$.
If $m$
is even, there are still two edges whose distance to the middle node
is $\sqrt{k^2+t^2}$.
So $\textsc{Weave}(m,n)$
is optimal whenever $n$
odd as its shortest edge has length $\sqrt{k^2+t^2}$.
If $n$
is even, there are two central nodes. For $m$
odd, the longest distance appears twice and is $\sqrt{t^2+k^2}$.
As the longest distance to a corner appears only once if $m$
is even, the second longest distance has to appear in any tour. Thus,
$\sqrt{k^2+(t-1)^2}$
is an upper bound, too.  For even $n$,
there is still a gap between the length of the shortest edge of the
tour from $\textsc{Weave}(m,n)$
and the current upper bound. The absolute value of the gap between
lower and upper bound is $\sqrt{t^2+k^2}-\sqrt{t^2+(k-1)^2}$
if $m$
is odd and $\sqrt{k^2+(t-1)^2}-
\sqrt{t^2+(k-1)^2}$ if
$m$
is even, respectively. But we can show with Lemma~\ref{thm:est_gap}
that this gap is smaller than the minimal distance in the grid.

\begin{lemma}
  For any regular grid with even $n$,
  the gap between the solution of $\textsc{Weave}(m,n)$
  and the upper bound of the Maximum Scatter TSP deduced from the
  distances between corner and central points is always smaller than
  one. 
  \label{thm:est_gap}
\end{lemma}
\begin{proof}
  As the upper bound for even $m$
  is smaller than the upper bound for odd $m$
  (recall that $\sqrt{(t-1)^2+k^2}
  < \sqrt{t^2+k^2}$), it suffices to prove $\sqrt{t^2+k^2}<
  \sqrt{t^2+(k-1)^2}+1$. We can square both sides of the inequality,
  because all the values are greater zero:
  $t^2+k^2
  \stackrel{}{>} t^2+(k-1)^2+2\sqrt{t^2+(k-1)^2}+1\Leftrightarrow 
  0\stackrel{}{>} -2k+2+2\sqrt{t^2+(k-1)^2} \Leftrightarrow 
  k-1\stackrel{}{>}\sqrt{t^2+(k-1)^2}$. 
  By observing that $t^2
  >0$ for
  $m>1$ on the right hand side, the strict inequality is proved. 
  So the value of the gap between the lower and the upper bound is
  always strictly smaller than one.
\end{proof}
If the grid is quadratic, the number of columns and rows has to be
both odd or both even. As $m=n$
and therefore $t=k$,
the gap for even $m$
and $n$
vanishes. So $\textsc{Weave}(m,n)$ is optimal for quadratic grids. 

When the number of rows is two, we claim that $\textsc{Weave}(2,n)$ is
optimal for all numbers of columns. For odd~$n$ this was already
proven. In the following we prove this for the case of even~$n$.
Assume that some tour achieves the upper bound, then each central node
would have to be connected to both of its farthest corners.  But in
this case -- as there are four central nodes -- the edges connecting
central and corner nodes violate the subtour elimination inequality: a
contradiction. Thus, the upper bound cannot be achieved by a feasible
tour, and each tour must employ at least one shorter edge incident to
a central node. The size of this new bound is $\sqrt{1+(k-1)^2}$.
Since this equals the lower bound achieved by $\textsc{Weave}(2,n)$,
the claim is proved.
  
Calculating a bound for the approximation factor by the division of
the approximate solution and the upper bound for the shortest edge, we
can summarize the quality 
of $\textsc{Weave}(m,n)$ in the following proposition.
\begin{proposition}
  $\textsc{Weave}(m,n)$ is optimal whenever $n$ is odd, or $m=2$, or $m=n$.\\
  For even $n$ and odd $m$ we have $\textsc{Weave}(m,n)$
  $\geq
  \underbrace{\sqrt{1-\frac{2k-1}{t^2+k^2}}}_{\alpha_{t,k}}\cdot$
  OPT(m,n).  For even $m,n$ we have $\textsc{Weave}(m,n)$
  $\geq
  \underbrace{\sqrt{1-\frac{2(k-t)}{(k-t)^2+2t(k-1)+1}}}_{\alpha_{t,k}}\cdot$
  OPT(m,n).
\end{proposition}
Since the approximation guarantee converges to one for increasing $m$
and $n$, we conclude the following:
\begin{corollary}
  In all cases, $\alpha \geq \frac{\sqrt{10}}{5}$. Moreover,
  $\lim_{k\rightarrow \infty} \alpha_{t,k} = 1$ for all 
  $t$ and \\
  $\lim_{t\rightarrow \infty,k\rightarrow \infty} \alpha_{t,k} = 1$.
  In this sense, $\textsc{Weave}(m,n)$ is asymptotically optimal.
\end{corollary}
Since $\textsc{Weave}(m,n)$ works by a fixed formula for each edge, it
is a linear-time algorithm. From this, Theorem~\ref{thm:main_thm}
follows.

\section{Conclusion}
\label{sec:4}

We have presented an asymptotically optimal linear-time algorithm
$\textsc{Weave}(m,n)$ for the Maximum Scatter TSP for the complete
graph on the points in the regular $m \times n$-grid in the plane.
This is the first polynomial time algorithm that solves an infinite
class of two-dimensional Maximum Scatter TSP in the plane to
optimality.  The complexity of the general Maximum Scatter TSP in the
plane remains open.


\begin{thebibliography}{99.}%
\bibitem{ACM} Arkin, E., Chiang, Y.-J., Mitchell, J.S.B., Skiena, S.S., Yang, T.-C.: On the Maximum Scatter TSP. SIAM~J.~Computing \textbf{29}, 515--544 (1999)
\bibitem{Chiang} Chiang, Y.-J.: New Approximation Results for the Maximum Scatter TSP. Algorithmica \textbf{41}, 309--341 (2005)
\end{thebibliography}
\end{document}